\documentclass[11pt,reqno]{amsart}

\usepackage{amsmath,amsthm,amssymb,comment,fullpage}
\usepackage{braket}
\usepackage{mathtools}

\usepackage{caption}
\usepackage{subcaption}
\usepackage{times}
\usepackage[T1]{fontenc}
\usepackage{mathrsfs}
\usepackage{latexsym}
\usepackage[dvips]{graphics}
\usepackage{epsfig}
\usepackage{amsmath,amsfonts,amsthm,amssymb,amscd}
\input amssym.def
\input amssym.tex
\usepackage{color}
\usepackage{hyperref}
\usepackage{url}
\newcommand{\bburl}[1]{\textcolor{blue}{\url{#1}}}

\usepackage{tikz}
\usepackage{tkz-tab}
\usepackage{tkz-graph}
\usetikzlibrary{shapes.geometric,positioning}

\newcommand{\burl}[1]{\textcolor{blue}{\url{#1}}}

\numberwithin{equation}{section}

\newtheorem{thm}{Theorem}[section]

\newtheorem{prop}[thm]{Proposition}

\theoremstyle{plain}

\newtheorem{corollary}[thm]{Corollary}
\newtheorem{definition}[thm]{Definition}

\newtheorem{theorem}[thm]{Theorem}
\newtheorem{problem}[thm]{Problem}

\newcommand\be{\begin{equation}}
\newcommand\ee{\end{equation}}
\newcommand\bee{\begin{equation*}}
\newcommand\eee{\end{equation*}}
\newcommand\bea{\begin{eqnarray}}
\newcommand\eea{\end{eqnarray}}
\newcommand\beae{\begin{eqnarray*}}
\newcommand\eeae{\end{eqnarray*}}
\newcommand\bi{\begin{itemize}}
\newcommand\ei{\end{itemize}}
\newcommand\ben{\begin{enumerate}}
\newcommand\een{\end{enumerate}}
\newcommand\bc{\begin{center}}
\newcommand\ec{\end{center}}
\newcommand\ba{\begin{array}}
\newcommand\ea{\end{array}}

\newcommand\frakfamily{\usefont{U}{yfrak}{m}{n}}
\DeclareTextFontCommand{\textfrak}{\frakfamily}

\newcommand{\hr}[1]{\href{#1}{\url{#1}}}

\newcommand{\ch}{\textnormal{ch}}
\newcommand{\textAnd}{\hspace{3mm}\textnormal{and}\hspace{3mm}}

\def\twobytwoMat(#1, #2, #3, #4){
    {
        \begin{bmatrix}
            {#1} & {#2}\\
            {#3} & {#4}
        \end{bmatrix}
    }
}

\def\twobyoneMat(#1, #2){
    {
        \begin{bmatrix}
            {#1}\\
            {#2}
        \end{bmatrix}
    }
}

\def\twobytwoDet(#1, #2, #3, #4){
    {
        \begin{vmatrix}
            {#1} & {#2}
            {#3} & {#4}
        \end{vmatrix}
    }
}

\newcommand\blfootnote[1]{%
  \begingroup
  \renewcommand\thefootnote{}\footnote{#1}%
  \addtocounter{footnote}{-1}%
  \endgroup
}

\title{Leslie Population Models in Predator-prey and Competitive populations: theory and applications by machine learning}

\author[Gilman, Miller, Son, Waheed, Wang]{Pico Gilman, Steven J. Miller, Daeyoung Son, Saad Waheed, Janine Wang}

\subjclass[2024]{92D25, 15A18  (primary) 37N25, 81Q10 (secondary)}

\keywords{Predator-Prey Model, Leslie Matrix, Migration Model, Quantum Operators, Machine Learning}

\date{\today}

\begin{document}
\maketitle

\blfootnote{
\begin{center}
\textit{Email addresses:}\texttt{~picogilman@gmail.com} (Pico Gilman),\texttt{~sjm1@williams.edu} (Steven J. Miller),\texttt{~ds15@williams.edu} (Daeyoung Son),\texttt{~sw21@williams.edu} (Saad Waheed), \texttt{~jjw3@williams.edu} (Janine Wang)
\end{center}
}

\begin{abstract}
We introduce a new predator-prey model by replacing the growth and predation constant
by a square matrix, and the population density as a population vector. The classical
Lotka-Volterra model describes a population that either modulates or converges. Stability analysis of such models have been extensively studied by the works of Merdan \cite{MeD09}, \cite{Mer10}. The new model adds complexity by introducing an age group structure where the population of each age group evolves as prescribed by the Leslie matrix.

The added complexity changes the behavior of the model such that the population either
displays roughly a exponential growth or decay. We first provide an exact equation that describes a time evolution, and use analytic techniques to obtain an approximate growth factor. We also discuss the variants of the Leslie model, i.e., the complex value predator-prey model and the competitive model. We then prove the Last Species Standing theorem that determines the dominant population in the large time limit.

The recursive structure of the model denies the application of simple regression. We discuss a machine learning scheme that allows an admissible fit for the population evolution of Paramecium Aurelia and Paramecium Caudatum. Another potential avenue to simplify the computation is to use the machinery of quantum operators. We demonstrate the potential of this approach by computing the Hamiltonian of a simple Leslie system.
\end{abstract}

\tableofcontents

\thanks{This work was partially supported by NSF Grant DMS2241623, Williams College, and The Finnerty Fund.  We are grateful to Professor Allison L. Gill from the Williams Biology Department for
her advice on gathering data for Section \ref{sec:machinelearning}.}


\section{Introduction}\label{sec:intro}

Leslie matrices describe the time evolution of a homogeneous population with multiple age groups. Consider a discrete-time evolution of a whale population that is comprised of three age groups. We let $a^{(i)}_n:\mathbb Z_{\rm pos} \rightarrow \mathbb R$ for $1 \ \leq \  i \ \leq \  3$ be the time-dependent populations of each age group; for example, $a_n^{(1)}$ is the population of newborns at time $n$. We define the population vector as
\begin{equation}
    \vec a_n \ := \ (a^{(1)}_n, a^{(2)}_n, a^{(3)}_n)^T.
\end{equation}
The total population is the sum of the populations of all age groups, or the sum of all entries in the population vector: $a^{(1)}_n + a^{(2)}_n + a^{(3)}_n$.

If we set the fertility rate of the whales to be $f \ > \  0$, constant across all age groups, and assume that the whales have a survival rate of 1, that is, they do not die from reasons other than old age, then we obtain a set of equations that describe the time evolution of the population:
\begin{align}
    a^{(1)}_{n + 1} \ &= \ f\cdot(a^{(1)}_n + a^{(2)}_n + a^{(3)}_n)
    \nonumber \\
    a^{(2)}_{n + 1} \ &= \ a^{(1)}_n \\
    a^{(3)}_{n + 1} \ &= \ a^{(2)}_n. \nonumber
\end{align}

This equation can be rewritten in matrix form. For this, we define the Leslie matrix $L$
\begin{eqnarray}
   L \ = \
    \begin{bmatrix}
        f & f & f \\
        1 & 0 & 0 \\
        0 & 1 & 0
    \end{bmatrix}, \nonumber
\end{eqnarray}
and we get that
\begin{eqnarray}
    \vec a_{n + 1} \ =\ L\vec a_n \ = \
    \begin{bmatrix}
        f & f & f \\
        1 & 0 & 0 \\
        0 & 1 & 0
    \end{bmatrix}\vec a_n,
\end{eqnarray}
which models our whale population. The advantage of using Leslie matrices is that the population vector at any given time can be expressed as a matrix power. If the population vector at time zero is $\vec a_0$, then the population vector at
time $n$ is
\begin{equation}
    \vec a_n  \ = \ L^n \vec a_0.
\end{equation}
The expression can be further simplified for faster computation. If the population vector $\vec a_0$ is an eigenvector of the Leslie matrix $L$ with an eigenvalue of $\lambda \in \mathbb R$, then we obtain
\begin{equation}
    \vec a_n  \ = \ \lambda^n \vec a_0,
\end{equation}
and the growth rate of the population is characterized by the eigenvalue $\lambda$.

The same technique used to describe homogenous populations can be applied to heterogenous populations, as is the case in a predator-prey model. For this model, let us now assume that whales consume plankton, which only has one age group.
We denote the plankton population by $b_n:\mathbb Z_{\rm pos} \rightarrow \mathbb R$\footnote{
In reality, the population is positive. Nonetheless, for a simple presentation, we choose the domain of $b_n$ to be over the reals. In practice, when the population reaches a negative value, the population is considered to be extinct.
} and introduce a predation rate $k \ > \  0$, a plankton population multiplier $m\ > \  0$, as well as a plankton fertility rate $F \ > \  0$. Writing the
new population vector as
\begin{equation}
    \vec p_n \ := \ (a^{(1)}_n, a^{(2)}_n, a^{(3)}_n, b_n)^T,
\end{equation}
we arrive at a new model
\begin{eqnarray}\label{eqn:motivatingModel}
    \vec p_{n + 1} \ :=\ \widetilde L\vec p_n \ = \
    \begin{bmatrix}
        f & f & f & m\\
        1 & 0 & 0 & 0\\
        0 & 1 & 0 & 0 \\
        -k & -k & -k & 1 + F
    \end{bmatrix}\vec p_n.
\end{eqnarray}

Depending on the parameters $(f, F, m, k)$, the model can either describe a situation where the predator population exhausts the prey population with too high a predation rate (itself eventually also becoming extinct due to starvation), or one where, with an appropriate predation rate, both populations grows. These two cases are illustrated in Figure \ref{fig:Mot12}.


\begin{figure}[htp]
    \centering
    \begin{subfigure}[b]{0.45\textwidth}
        \includegraphics[width=\textwidth]{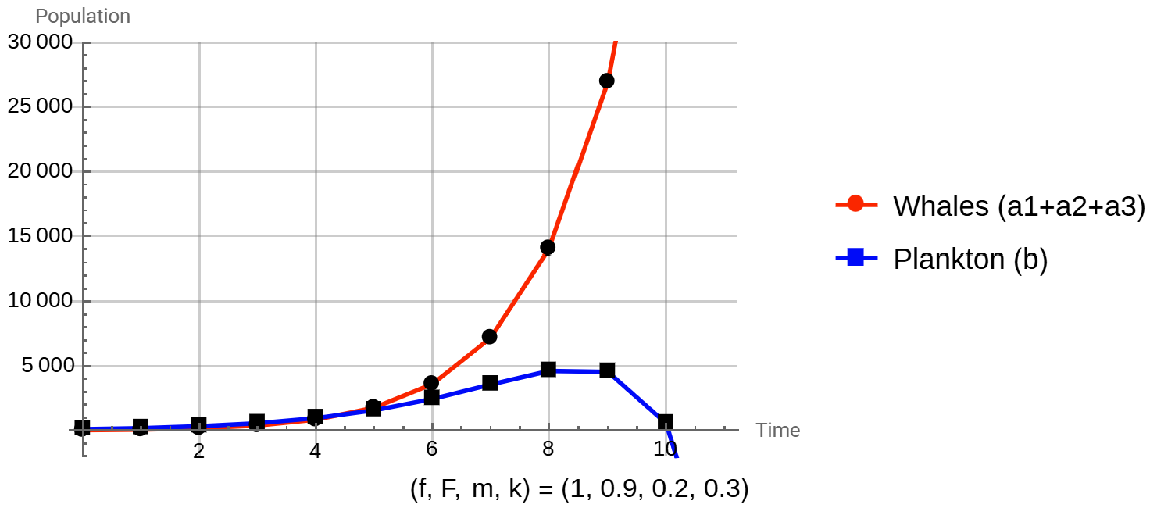}
        \caption{High predation and prey exhaustion}
        \label{fig:Mot2}
    \end{subfigure}
    \hfill
    \begin{subfigure}[b]{0.45\textwidth}
        \includegraphics[width=\textwidth]{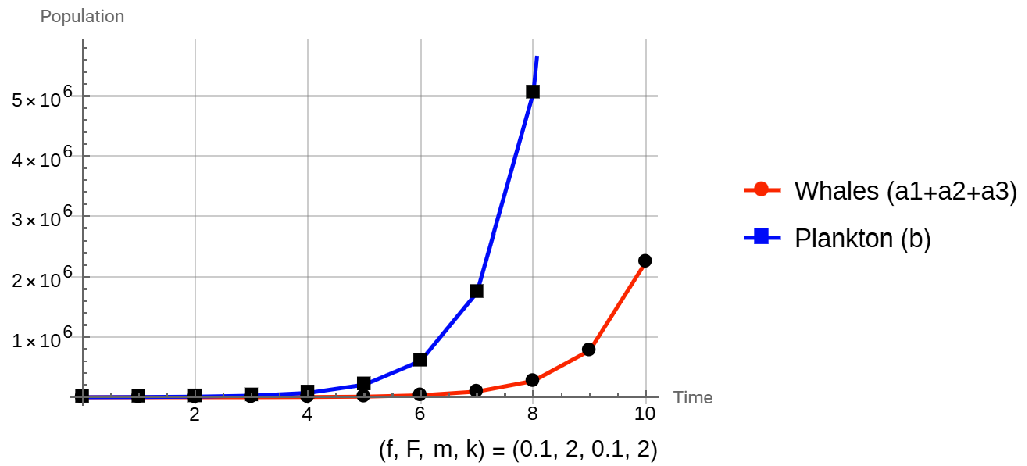}
        \caption{Low predation and mutual population growth}
        \label{fig:Mot1}
    \end{subfigure}
    \caption{Plot of model \ref{eqn:motivatingModel} for varying parameters.}
    \label{fig:Mot12}
\end{figure}

Not all parameters predict a realistic population evolution. For example, for certain values of $(f, F, m, k)$, \ref{eqn:motivatingModel} displays an oscillatory behavior where the population is described as a negative number.

\begin{figure}[htp]
    \centering
    \includegraphics[width=0.8\textwidth]{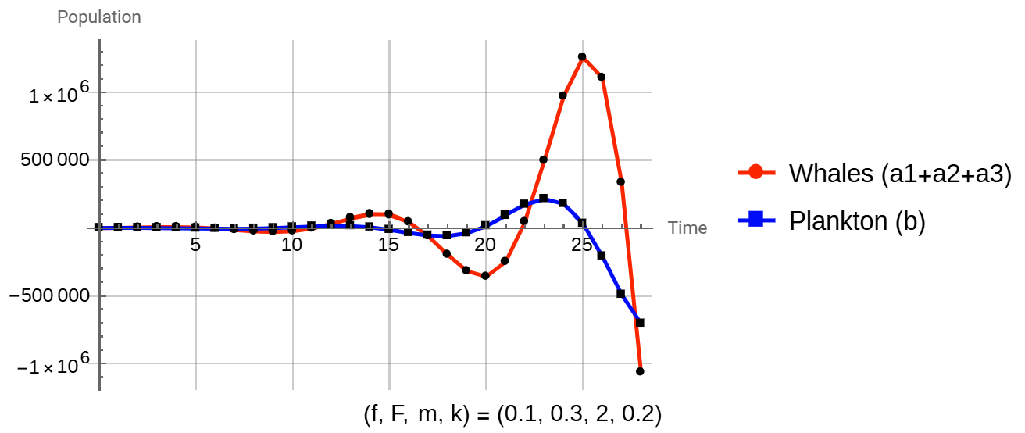} 
    \caption{Figure of oscillatory population with high predation.}
    \label{fig:example}
\end{figure}

In the following sections, we begin by defining and studying the eigenvalues of a simple
Leslie matrix. Furthermore, we show that a population model with nonsimple Leslie matrices can be approximated by a simplified version with simple Leslie matrices.  (Section \ref{sec:single})
. We then introduce the Leslie predator-prey model for real and complex values. The solution for the complex model follows nicely from these results.
We also introduce the competitive population model where each population consumes the other to promote growth of their own species.
We present a closed-form formula for
the population of the predator-prey model with real value
using a generating function approach (Section \ref{sec:PPmodels}). The complexity of the closed-form formula motivates us to study the model for a
small number of age groups, i.e., one age group for each predator and prey
(Section \ref{sec:scalarModel}). Next, using the observations made in Section \ref{sec:single}, we provide an
asymptotic growth rate for the complex model and prove the last-species standing
theorem for the competitive model (Section \ref{sec:lastSpecies}).

To test the presented theory of Leslie population models, we used the competitive model to explain the population evolution Paramecium Aurelia and Paramecium Caudatum. The recursive structure denies the application of standard regression techniques. A Machine Learning Scheme developed to generate an admissible fit is presented, along with a fit data with statistic significance (Section \ref{sec:ML}).

Finally, the complex predator-prey model motivates our study to consider
the use of quantum mechanics
to describe the population. We investigate a specific case of population evolution and compute the 
Hamiltonian \footnote{For more information on the standard theory, refer to \cite{Bag19}.} of the system under the assumption that the population obeys the time-dependent Schrodinger equation.

\section{Single Species Population}\label{sec:single}

\subsection{Definition of Simple Leslie Matrices and the Lotka-Euler Equation}

As described in Section \ref{sec:intro}, Leslie matrices characterize changes in a species' population with different age groups given its survival and fertility rates. We focus on a specific class of Leslie matrices with a fixed fertility rate $f$ and a survival rate 1.

\begin{definition}[Leslie Matrices]
    \label{LeslieDef}

    Suppose $N \in \mathbb{Z}^+$ is the number of age groups and $f_1, \dots, f_N \in \mathbb R$ be fertility rates of each age group. A simple Leslie matrix that
    characterizes the population evolution is defined as follows:

    \[
        (L_{f_1, \dots, f_N})_{ij} \ =\ \begin{cases}
            f_i & (i \ = \  1)\\
            1 & (i \neq 1 \wedge j\ = \ i+1)\\
            0 & \rm {otherwise}\\
        \end{cases}
    \]
        or, writing the matrix out,
    \[(L_{f_1, \dots, f_N}) \ = \
    \begin{bmatrix}
        f_1 & f_2& \cdots & f_{N - 1} & f_N\\
        1 & 0 & \cdots & 0 & 0 \\
        0 & 1 & \cdots & 0 & 0\\
        \vdots & \vdots & \ddots & \vdots & \vdots\\
        0 & 0 &\cdots & 1 & 0
    \end{bmatrix}.
    \]

    If the fertility rate $f_1, f_2, \dots, f_N$ are constantly equal to some fixed fertility rate $f$, then we say that the Leslie matrix is \textbf{simple}. Also, for a non-simple Leslie matrix with a variable fertility rate, we define the simplified Leslie matrix $L_f$ as follows.
    \[L_f \ := \
    \begin{bmatrix}
        f & f& \cdots & f & f\\
        1 & 0 & \cdots & 0 & 0 \\
        0 & 1 & \cdots & 0 & 0\\
        \vdots & \vdots & \ddots & \vdots & \vdots\\
        0 & 0 &\cdots & 1 & 0
    \end{bmatrix}.
    \]
    The quantity $f$ is the average fertility rate, which is defined as
    \[
        f \ := \ \frac {\sum_{i \ = \  0}^N f_i} N.
    \]

\end{definition}

The population at time $n$ can be modeled by a tuple of real numbers, which we write out as a vector.

\begin{definition}[Time evolution of a single population]
Denote the population vector as $\vec p_n \in \mathbb R^N$.
Given an initial population vector $\vec p_0 \ := \ \vec v$, the population at time $n$ is given by
\begin{equation}
    \vec p_n \ =\ (L_{f_1, \dots, f_N})^n \vec v
\end{equation}
\end{definition}

We present the following theorem, which allows us to focus our
analysis on simple Leslie matrices.

\begin{theorem} [Approximating general Leslie with simplified Leslie]
    Assume that the fertility rate $f_N$, i.e., the fertility rate of the oldest
    age group, is less than the average fertility rate. For time $n$ sufficiently large, the population vector for a model
    involving the general Leslie matrix with varying fertility rate can be approximated using simplified Leslie matrix using the following formula.
    \begin{equation}\label{eqn:simpleApprox}
        (L_{f_1, \dots, f_N})^n \vec v\ \approx \
        (L_f)^n \left[n\left(\sum_{k \ = \  1}^N v_k(f_k/f -1)\right)\vec e_N  + \vec v\right] \ := \
        (L_f)^n \left(n\xi_{v, f}\vec e_N  + \vec v\right)
    \end{equation}\footnote{We write $\xi_{v, f}$ to highlight that the quantity is dependent on the initial population vector $\vec v$ and the fertility rates $f_1,\dots, f_N$}
\end{theorem}

\begin{proof}
It is straightforward verify the following matrix identity.
\begin{equation}\label{eqn:invLL}
L_f^{-1}L_{f_1, \dots, f_N} \ = \ I + C
\end{equation}
The matrix $C$ is a square matrix of order $N$ with $N-1$ zero columns with
one nonzero column.
\begin{equation}
(C)_{ij} \ = \ \begin{cases}
f_i/f - 1 &(i \ = \  N)\\
0 & (i \ < \  N)
\end{cases}
\end{equation}
Also, since $C$ has $N-1$ zero rows, for any integer $k \ > \  1$, we observe that $C^k \ = \ (f_N/f - 1)^k E_N$ where
\begin{equation}
(E_N)_{ij} \ = \ \begin{cases} 1 & (i \ = \  N, j \ = \  N) \\
0 & \rm otherwise.
\end{cases}
\end{equation}
Take power of $n$ to the both sides of \eqref{eqn:invLL} and invoke the Binomial Theorem.
\begin{equation}
(L_f^{-1}L_{f_1, \dots, f_N} )^n \ = \
(I + C)^n \ = \  \sum_{k \ = \  0}^N \binom{n}{k} C^k
\ = \ I + nC + \sum_{k \ = \  2}^N \binom{n}{k} C^k
\end{equation}
Use our observation on the powers of $C$.
\begin{equation}
\begin{split}
(L_f^{-1}L_{f_1, \dots, f_N} )^n & \ =
\ I - E_N + nC - (f_N/f - 1) n E_N
+ E_N \sum_{k \ = \  0}^N \binom{n}{k} (f_N/f - 1)^k \\
& \ = \ I - E_N + nC - (f_N/f - 1) n E_N
+ E_N (f_N/f)^n\\
& \approx \ I - E_N + nC - (f_N/f - 1) n E_N \\
\end{split}
\end{equation}
Where the second equality follows from the binomial theorem and the approximation from the assumption that $f_N \ < \  f$. Multiply both sides by the initial population vector $\vec v$:
\begin{equation}
(L_f^{-1}L_{f_1, \dots, f_N} )^n\vec v \approx \
\vec v + n \vec e_N \left(\sum_{k \ = \  1}^N v_k(f_k/f -1)\right)
\end{equation}
The vector $e_N$ is the elementary basis vector, i.e.,
\begin{equation}
    (\vec e_N)_i \ = \ \begin{cases}
    1 & (i \ = \  N) \\
    0 & \rm otherwise.
    \end{cases}
\end{equation}
Finally, multiplying both sides by $(L_f)^{-n}$ to the left yields the desired result.
\end{proof}

The implication of \eqref{eqn:simpleApprox} is that for sufficiently large time $n$,
the major contribution comes from the term $\xi_{v, f}\vec e_N$, allowing an even cruder approximation
\begin{equation}\label{eqn:simpleApprox}
        (L_{f_1, \dots, f_N})^n \vec v\ \approx \
        (L_f)^n (n\xi_{v, f}\vec e_N).
\end{equation}
The absolute error of this approximation will be large. However, when considering the relative growth rate between two different populations, which we will be mainly concerned about, the approximation suffices.

To demonstrate this, suppose we wish to compute the proportion of the total populations of two different populations. Suppose the first population has a initial population vector of $\vec v$ and $N$ age groups with fertility rates $f_1, \dots, f_N$. Let the second population to have an inital population vector of $\vec w$ and $N$ age groups with fertility rates $g_1, \dots, g_N$. The total population vector of each population is the $L^1$ norm of the following two vectors, which are
\begin{equation}
 (L_{f_1, \dots, f_N})^n \vec v \hspace{5mm}\textnormal {and}  \hspace{5mm} (L_{g_1, \dots, g_N})^n \vec w.
\end{equation}
Thus the proportion of the two population is
\begin{equation}
\frac{|(L_{f_1, \dots, f_N})^n \vec v|_1}{|(L_{g_1, \dots, g_N})^n \vec w|_1} \ = \
\frac{
|(L_f)^n n \xi_{v, f} \vec e_N + (L_f)^n \vec v|_1
}{
|(L_g)^n n \xi_{w, g} \vec e_N + (L_g)^n \vec w|_1
}
\ = \ \frac{
|(L_f)^n \xi_{v, f} \vec e_N |_1
}{
|(L_g)^n \xi_{w, g} \vec e_N |_1
}
+ O\left(\frac 1 n \right).
\end{equation}

Later, we verify that the eigenvalue with the largest modulus is a positive real value. We call this eigenvalue the
\textbf{dominant eigenvalue}.
The dominant eigenvalue of this matrix describes the asymptotic behavior of the population. To begin our discourse,  we compute the matrix's characteristic equation and find its roots.

\begin{theorem}[Lotka-Euler Equation]
    \label{LEeq}
    The characteristic equation of a simple Leslie matrix $L_f$ of
    order $N\ \geq \ 1$ is
    \[
        \ch_N(x) \ = \ x^N - f(x^{N-1} + \cdots + x + 1)
    \]
    which, using the geometric series formula, can be simplified to
    \[
        x^{N} - f \frac {x^N - 1} {x - 1}.
    \]
\end{theorem}

\begin{proof} We induct on $N$. It is trivial to see that the equation holds for
$N \ = \  1$
. For the inductive step, consider $N \ > \  1$. We write out the
characteristic polynomial as a determinant expansion:

\[
    \ch_{N + 1}(x) \ := \ \textrm{det}(xI - L_f)
    \ = \
    \begin{vmatrix}
        x - f & -f& \cdots & -f & -f \\
        -1 & x & \cdots & 0& 0 \\
        0 & -1 & \cdots & 0 & 0\\
        \vdots & \vdots &\ddots & \vdots & \vdots\\
        0 & 0 & \cdots & -1 & x
    \end{vmatrix} .
\]

Expanding this with respect to the last column yields
\[
    \ch_{N + 1}(x) \ = \
    (-f)(-1)^N(-1)^N + x\ch_{N}(x).
\]
By the inductive hypothesis, we have
\[
    \ch_{N + 1}(x) \ = \
    -f + x\left(
    x^N - f(x^{N-1} + \cdots + x + 1)
    \right)
    \ = \
x^{N + 1} - f(x^{N} + \cdots + x + 1),
\]
which concludes the proof.
\end{proof}

We focus on the case where the simple Leslie matrix describes a growing population. We comment that the roots of $\ch_N(z)$ are exactly the
eigenvalues of $L_f$. As a corollary of the theorems in Appendix \ref{app:A}, we present the following result.

\begin{corollary}
    The eigenvalues of a simple Leslie matrix $L_f$ all
    have eigenvalues that have a modulus strictly less than 1, assuming $f \ < \ 1/N$. Otherwise, if $f \ \geq \ 1/N$, then there exists a unique positive real eigenvalue, greater than or equal to 1, that describes the asymptotic growth of the population.
    In other words, there exists $\lambda_{\max} \in [1, \infty) $ such that for any $\vec v \in \mathbb R^n$, there exists some $\vec w \in \mathbb R^n$ such that
    \begin{align}
        (L_f)^n \vec v \ \approx \ (\lambda_{\max})^n \vec w
    \end{align}
    for large enough $n \in \mathbb Z^+$.
    We call such eigenvalue of $L_f$ as the dominant eigenvalue.
\end{corollary}

We focus on populations that grow in the long term period. Hence, we assume $f \ \geq \ 1/N$ for most cases.

\section{The Leslie Predator-Prey Model}

\subsection{The Classic Lotka-Volterra Predator-Prey Model}
Let $x(t)$ and $y(t)$ be continuous functions
that describe the respective densities of prey and predator populations. That is, both $x$ and $y$ have ranges within the interval $[0, 1]$. The classic predator-prey model is described by a system of differential equations:
\begin{eqnarray}
    \frac{dx}{dt} & \ = \  &
    rx(1-x) - axy \nonumber\\
    \frac {dy}{dt} & \ = \ &
    ay(x-y),
\end{eqnarray}
$r, a \ > \  0 $ are reproductive ratio and predation ratio respectively.
Studies of the classic model focus on finding the conditions for which the system reaches stability. In \cite{Mer10}, Merdan explores a
similar system that accounts for the Allee effect, where the population growth is diminished when the population size is small. The model is described by the equations below:
\begin{eqnarray}
    \frac {dx}{dt} &\  =  \ & r\alpha(x) x (1-x) - axy \nonumber \\
    \frac {dy}{dt} & \ = \ & ay(x - y).
\end{eqnarray}
The term $\alpha(x) \ := \ x/(\beta + x)$ captures the Allee effect. Merdan shows that under the condition
\begin{eqnarray}
    r - \alpha \beta  \ > \ 0,
\end{eqnarray}
the population converges to a
positive stable state:
\begin{eqnarray}
    (x_*, y_*) \ = \ ((r - \alpha\beta)/(a + r), (r-\alpha\beta)/(a + r)).
\end{eqnarray}

\subsection{The Predator-Prey Model with Leslie Matrices}\label{sec:PPmodels}

We wish to account for different age groups in the predator and prey populations. Hence, we replace population density, which was previously a scalar function, by a vector. We also replace the previous reproductive and predation ratios by Leslie matrices.

\begin{definition}[Leslie Predator-Prey Model]
Let $\vec \alpha_n$, $\vec \beta_n \in \mathbb R_{\rm pos}^N$ be the population vectors
for the predator and prey species at time $n$
Let both populations have $N$ different age groups, resulting in the following population vectors:
\begin{eqnarray}
    \vec \alpha_n & \ = \ & (\alpha_n^{(1)}, \dots, \alpha_n^{(N)})^T
    \nonumber \\
    \vec \beta_n & \ = \ & (\beta_n^{(1)}, \dots, \beta_n^{(N)})^T.
\end{eqnarray}
 The population
vectors are defined by the following
system of matrix difference equations:
\begin{eqnarray}
    \vec \alpha_{n + 1} & \ = \ & L_a \vec \alpha_n + k m \vec \beta_n \nonumber \\
    \vec \beta_{n + 1} &\ = \ & L_b \vec \beta_n - k \vec \alpha_n.
\end{eqnarray}
The constants $k$ and $m$ are respectively the predation and nurturing ratios, both greater than zero \footnote{
The predation ratio describes the amount of prey consumed
by each predator. The nurturing ratio describes the population boost that comes from predation.
}
. We set the predator and prey populations to be the sum of their vector entries. Symbolically, we let $P_{a, n}, P_{b, n}$ denote the predator and prey populations, given by the following sums:
\begin{equation}
    P_{a, n} \ = \ \sum_{k \ = \  1}^{N} \alpha^{(k)}_{n}  \hspace{3mm} \textrm{and} \hspace{3mm}
    P_{b, n} \ = \ \sum_{k \ = \  1}^{N} \beta^{(k)}_{n}.
\end{equation}
\end{definition}

We assume that the dominant eigenvalue of $L_\alpha$ does not have a dominant eigenvalue and $L_\beta$ has a dominant eigenvalue.
In other words, the predator population decays in absence of prey and the prey population explodes in absence of predators.

Furthermore, the populations are fixed to be non-negative. If a population reaches zero
at some time $n \in \mathbb{Z}^+$,
we say the species has gone extinct. Notice that if the predation rate is too high, the prey population
will be exhausted and subsequently the predator
population will also become extinct. On the other hand, if the predation rate is too low, the predator population will be unable to sustain
itself and will equally become extinct. Hence, it is natural to ask the following question.

\begin{problem}[Optimal Predation Strategy]
What range of the real value $k$ guarantees exponential growth of the predator population? Moreover, what value of $k$ ensures maximum growth?
\end{problem}

The real-valued Leslie Predator-Prey model motivates us to study a complex-valued model.
By multiplying $i \ = \  \sqrt{-1}$ to one of the summands can
be considered as a time-delay in the population change.

\begin{definition}[Complex Leslie Predator-Prey Model]\label{thm:complexModel}
Let $\alpha_n, \beta_n \in \mathbb R_{\rm pos}^N$ for $n \ = \  0$ and $\alpha_n, \beta_n \in \mathbb C^N$ for $n \ > \  0$ be the population vectors
of the predator and prey species at time $n$:
\begin{eqnarray}
    \vec \alpha_n & \ = \ & (\alpha_n^{(1)}, \dots, \alpha_n^{(N)})^T
    \nonumber \\
    \vec \beta_n & \ = \ & (\beta_n^{(1)}, \dots, \beta_n^{(N)})^T.
\end{eqnarray}
 The population
vectors are defined by the following
system of matrix difference equations:
\begin{eqnarray}
    \vec \alpha_{n + 1} & \ = \ & iL_a \vec \alpha_n + k m \vec \beta_n \nonumber \\
    \vec \beta_{n + 1} & \ = \ & iL_b \vec \beta_n - k \vec \alpha_n,
\end{eqnarray}
where $k$ and $m$ are respectively the predation and nurturing ratios, both greater than zero. We again set the predator and prey populations, $P_{a, n}$ and $P_{b, n}$, to be the sum
of their vector entries. In symbols, we have
\begin{equation}
    P_{a, n} \ = \ \left\|\sum_{k \ = \  1}^{N} \alpha^{(k)}_{n} \right\|
    \hspace{3mm} \textrm{and} \hspace{3mm}
    P_{b, n} \ = \ \left\|\sum_{k \ = \  1}^{N} \beta^{(k)}_{n} \right\|
\end{equation}
\end{definition}


Experimentally speaking, for the complex model, the population grows almost surely unless the predation rate $k$ is zero. The natural question to ask for this model is therefore the following:

\begin{problem}[Modeling Predator Growth]
    What is the growth rate of the predator population as $n \rightarrow \infty$?
\end{problem}


By elementary substitutions, we obtain the following proposition.

\begin{prop}[Coupled 1st Order to 2nd Order
]\label{thm:to2ndOrd}
Assuming both prey and predator populations are non-extinct within a given period of time, the populations satisfy the following second order
difference equations:
\begin{eqnarray}
\vec \alpha_n & \ = \ & (L_a + L_b) \vec \alpha_{n - 1} - L_b L_a \vec \alpha_{n -2} - m k^2 \vec \alpha_{n - 2} \nonumber
\\
\vec \beta_n & \ = \ & (L_b + L_a) \vec\beta_{n - 1} - L_a L_b \vec\beta_{n -2} - m k^2 \vec\beta_{n - 2}.
\end{eqnarray}
For the complex model, we have equivalently
\begin{eqnarray}
\vec \alpha_n & \ = \ & i(L_a + L_b) \vec \alpha_{n - 1} - L_b L_a \vec \alpha_{n -2} - m k^2 \vec \alpha_{n - 2} \nonumber
\\
\vec \beta_n & \ = \ & i(L_b + L_a) \vec\beta_{n - 1} - L_a L_b \vec\beta_{n -2} - m k^2 \vec\beta_{n - 2}.
\end{eqnarray}
\end{prop}

The coupled second order differences equation can be solved using generating functions under the assumption that the Leslie matrices of the two populations are constant multiples of each other

\begin{theorem} [Generating Function of the Predator Population in the Real Case]
    Let $\vec \alpha_n$ be the predator population vector in the real Leslie predator-prey model where
    $L_a \ = \ \rho L$ and $L_b \ = \ L$.
    The generating function of $\vec \alpha_n$ is
    \begin{equation}\label{eqn:genFunc}
        G(x) \ = \
\left[\left(\rho L + m k^2 I - (\rho + 1) L x\right)\vec \alpha_0 + (\rho L + m k^2 I) x \vec  \alpha_1\right]\left[x^2 I - x(\rho + 1) L + \rho L^2 + m k^2 I\right]^{-1}.
    \end{equation}
\end{theorem}

\begin{proof}
    From the recurrence relation provided in Proposition \ref{thm:to2ndOrd},
    we have the following identity:
    \begin{equation}
        \left[x^2 - x (\rho + 1) L + (\rho L^2 + mk^2)\right]G(x) \ = \
        -(\rho + 1) L \vec \alpha_0 x
        +(\rho L^2 + mk^2) \vec \alpha_0
        + (\rho L^2  + m k^2) \vec \alpha_1 x.
    \end{equation}
    This identity can be verified by substituting $G(x)$ and imposing the appropriate conditions on $\alpha_n$. The expansion on the left-hand side has residues for terms that have an $x$ power less than or equal to 2. Solving for $G(x)$ yields the desired result.
\end{proof}

Using partial fraction decomposition, it is possible to obtain a closed form expression for $\vec \alpha_n$.

\begin{theorem} [Formula for $\vec \alpha_n$]
    When $n \ > \  0$,
    \begin{equation}
    \vec \alpha_n \ = \
    \frac {
(L^2 \rho + k^2 m)^{n - 2}
    } {\sqrt D}
    \left[
        \left(
            (k^2m + L^2 \rho) \vec \alpha_1 - L(1 + \rho)
\left(
            k^2m + L^2\rho
        \right)
            \vec \alpha_0
        \right)\delta_{n - 1}
        + \left(
            k^2m + L^2\rho
        \right)^2\vec \alpha_0 \delta_{n}
    \right],
    \end{equation}
    where $D$ is defined by
    \begin{equation}
        D \ = \ L ^2 (1 + \rho)^2 - 4 (mk^2 + \rho L^2),
    \end{equation}
    and the sequence $\delta_n$ is defined by
    \begin{equation}
        \delta_n \ = \
        \left(
            \frac 1 2
        \right)^n
        \left(
                L^2 \rho + k^2 m
        \right)^{-n}
        \left[
            \sum_{\substack{l \ = \  1 \\ l \textnormal{ odd}}}^{n + 1}
            \binom{n + 1}{l}
            [L(1+\rho)]^{t + 1 - l}(\sqrt D)^l
        \right].
    \end{equation}
\end{theorem}

\begin{proof}
    The derivation follows by applying partial fraction decomposition to the previous
    generating function
    in Equation \ref{eqn:genFunc}.
\end{proof}

Though the theorem provides a closed-form expression for the predator population, the complexity of the formula poses difficulties in determining the optimal predation rate for maximal growth.

\subsection{Real-Valued Predator-Prey Model with Scalar $L$}\label{sec:scalarModel}

The following three propositions model
the predator and prey populations when the dimension of the Leslie matrix is 1; that is, the population growth is characterized by an exponential of a scalar without interaction.
 To emphasize their scalarity, we write $l_a \ < \  1$ and $l_b \ > \  1$ instead of $L_a$ and $L_b$.

\begin{theorem}[Eigenvalues of the Companion Matrix]\label{thm:simpleEval}
Using Proposition \ref{thm:to2ndOrd}, we write the companion matrix that describes the two populations:
\begin{equation}\label{eqn:1DcompMat}
    \twobytwoMat(
        l_a + l_b, -l_a l_b -k^2 m, 1, 0
    ).
\end{equation}
The eigenvalues of this matrix are purely real if and only if
\begin{equation}
    k \ \leq\  \frac {l_a - l_b} {2\sqrt{m}}.
\end{equation}
Otherwise, the eigenvalues are complex conjugates of each other.
\end{theorem}

\begin{proof}
    The characteristic equation of the companion matrix is
    \begin{equation}
        \lambda ^2 - (l_a + l_b) \lambda + k^2m + l_a l_b.
    \end{equation}
    For both eigenvalues to be purely real, the discriminant $D$ of this polynomial must be nonnegative:
    \begin{equation}
        \frac D 4 \ :=\ \frac {(l_a + l_b)^2} 4 - k^2m + l_a l_b \ \geq \ 0.
    \end{equation}
    Using elementary algebra, we obtain
    \begin{equation}
    k \ \leq\ \frac {l_a - l_b} {2\sqrt{m}}.
    \end{equation}
    Otherwise, if $D/4 \ < \  0$, the eigenvalues have an imaginary part, and the two eigenvalues are complex conjugates of each other.
\end{proof}

\begin{theorem}[Exponential Growth of
Population for Small Predation Rate]
The following condition guarantees that neither predator nor prey populations vanish as $n \rightarrow \infty$:
\begin{equation}
    k \ \leq \ \sqrt{
        \frac {(1 - l_b)(l_a - 1)} {m}
    }.
\end{equation}
\end{theorem}

\begin{proof}
    We assume that the discriminant of the companion matrix (Equation \ref{eqn:1DcompMat}) is a non-negative real value.
    Then the dominant eigenvalue must be
    \begin{equation}
        \frac {l_a + l_b} 2 + \frac {\sqrt D} 2,
    \end{equation}
    which must be greater or equal to 1 for both of the population to not vanish.
\end{proof}

\begin{theorem}[Extinction in the Case of Complex Eigenvalues]
If
\begin{equation}\label{eqn:simpleVanish}
    k \ > \  \frac {l_a - l_b} {2\sqrt{m}},
\end{equation}
then the population 
is guaranteed to go extinct.
\end{theorem}
\begin{proof}
    It follows trivially that condition (\ref{eqn:simpleVanish})
    implies the dominant eigenvalue is complex. Furthermore, the real part of the root is $(l_a + l_b)/2$, which is guaranteed to be positive. Let the two eigenvalues of the companion matrix be $\gamma$ and $\bar \gamma$, with
    \begin{equation}
    \gamma \ = \ r e^{i \theta} \hspace{3mm} \textrm{and} \hspace{3mm} \ \bar \gamma \ = \ r e^{-i \theta},
    \end{equation}
    where $r \ > \  0$ and $\theta \in (0, \pi/2)$. By Proposition \ref{thm:to2ndOrd}, we note that the predator population at time $n$ can be written as
    \begin{equation}\label{eqn:gammaPopulation}
    \alpha_n \ = \ \nu_1 \gamma^n + \nu_2 (\bar \gamma)^n.
    \end{equation}
    We also observe that the populations $\alpha_0$ and $\alpha_1$ can be assumed to take positive real values. If $\alpha_1 \ \leq \  0$, then the population has gone extinct at time 1.
    Equation \ref{eqn:gammaPopulation} for $n \ = \  0, 1$ is
    \begin{eqnarray}
    \alpha_0 & \ = \ & \nu_1 + \nu_2 \nonumber\\
    a_1 & \ = \ & \nu_1\gamma + \nu_2\bar \gamma.
    \end{eqnarray}
    Since $\alpha_0, \alpha_1 \ > \  0$, we deduce that $\nu_1 \ = \ \nu_2 \ := \ \nu/2 \ > \  0$. Finally, we rewrite the population at time $n$:
    \begin{equation}
    \alpha_n \ = \ \Re (\nu \gamma ^n) \ = \  \nu r^n \cos(n\theta).
    \end{equation}
    We know that $\theta \in (0, \pi / 2)$ and thus, there exists an integer $n$ such that $\cos(n\theta) \ < \  0$, which finishes showing that the predator population must go extinct.
\end{proof}
\hfill \qed

\subsection{The Complex-Valued Leslie Predator-Prey Model with $L_a \ = \ \rho L_b$}

To solve the second order matrix recurrence related to the predator-prey model, we solve a characteristic equation whose coefficients are matrices. Since the only matrices involved in this equation are $I$ and $L_\beta$ which commute, we can use the quadratic equation.

\begin{theorem}[Dominant Eigenvalue in the General Case]\label{eqn:MaxEvalGen}
    The population vector of the predator species in (\ref{LeslieDef}) can be characterized by
    \begin{equation}
    \vec \alpha_n \ =\ \Lambda_1^{n}  \vec v_1+ \Lambda_2^{n}\vec v_2
    \end{equation}
    for vectors $\vec v_1$ and $\vec v_2$, the eigenvectors of the system.
    . The
    growth of the predator population
    is dominated by the dominant eigenvalue
    of $\Lambda_1$. We denote the dominant eigenvalue of $L_b$ by $\lambda_{\max}$ and that of $\Lambda_1$ by $\Lambda_{\max}$. Then $\Lambda_{\max}$ has the following modulus:
    \begin{equation}
        \|\Lambda_{\max} \| \ =\   \frac {
            (\rho + 1)\lambda_{\max} + \sqrt{(\rho + 1)^2 \lambda_{\max}^2 + 4mk^2}
        }{2}.
    \end{equation}.
\end{theorem}
\begin{proof}
    It is possible to solve for $\Lambda_1, \Lambda_2$ directly.
We wish to find a matrix $\Lambda$ such that
\begin{equation}
\Lambda^2 - i(\rho + 1) L_b \Lambda + \rho L_b^2+ mk^2 I \ = \ 0.
\end{equation}
    Applying the quadratic formula yields
    \begin{eqnarray}
        \Lambda_1 &\ = \ &
        \frac {
        (1 + \rho) L_b^2 +
        \sqrt{
        (1-\rho)^2L_b^2 + 4mk^2
        }
        } 2 i \nonumber \\
 \Lambda_2 & \ = \ &
        \frac {
        (1 + \rho) L_b^2 -
        \sqrt{
        (1-\rho)^2L_b^2 + 4mk^2
        }
        } 2 i.
    \end{eqnarray}
    Note that the magnitude of $\Lambda_1$ is greater than that of $\Lambda_2$.
    We approximate the population of the predator species at the limit as $n \rightarrow \infty$:
    \begin{eqnarray}
        P_{a, n} \ = \ \|\vec a_n\| \ = \
        \|\Lambda_1\|^n \|\vec v_1\| +
 \|\Lambda_2\|^n \|\vec v_2\|
 \ \approx  \
\|\Lambda_1\|^n \|\vec v_1\|.
    \end{eqnarray}
    It remains to show that the vector $\vec v_1$ is nonzero. Let us assume for contradiction that $\vec v_1 \ = \ (0, \dots, 0)^T$. Then we can write the predator populations at times 0, 1 as
    \begin{equation}
    \vec \alpha_0 \ = \ \vec v_2 \hspace{3mm} \textrm{and} \hspace{3mm} \vec \alpha_1 \ = \ \Lambda_2 \vec v_2,
    \end{equation}
    which indicates that
    \begin{equation}
    \vec \alpha_1 \ = \ \Lambda_2 \vec \alpha_0.
    \end{equation}
    Since $\Lambda_2$ is purely imaginary, $\alpha_1$ is therefore also purely imaginary.
    However, the initial condition of the model in Definition \ref{thm:complexModel}
    dictates that each entry of $\vec \alpha_0$ and $\vec \beta$ is positive and real and that
    \begin{equation}
    \vec \alpha_1 \ = \ iL_{\alpha} \vec \alpha_n + km \vec\beta_n.
    \end{equation}
    Therefore, $\alpha_1 $ cannot be purely imaginary and we arrive at a contradiction.
\end{proof}

\section{The Competitive Model}
We can slightly modify one of the signs in the previous Leslie predator-prey model and
study the following system. Suppose there exist two populations with the same growth matrix $L$. We assume the two populations are non-vanishing without interaction; that is, $f \ \geq \ 1/N$ and $L \ = \ L_f$ has a dominant eigenvalue.

\begin{definition}[Leslie Competitive Predator-Prey  Model]
Let $\vec \alpha_n$, $\vec \beta_n$ be the population vectors
of the predator and prey species at time $n$. The competitive model is defined by the following
system of matrix difference equations:
\begin{eqnarray}\label{eqn:conditions}
    \vec \alpha_{n + 1} & \ = \ & \max(L \vec \alpha_n - k m \vec \beta_n, \vec 0) \\
    \vec \beta_{n + 1} & \ = \ & \max(L \vec \beta_n - k \vec \alpha_n, \vec 0) \nonumber,
\end{eqnarray}
where $k$ and $m$ are respectively the interaction and competitive advantage ratios, both between $0$ and $1$. The interaction ratio describes how much interaction, i.e.,, how much casualties are incurred by competition. The competitive advantage ratio describes the competitive ratio of species $\beta$ over $\alpha$.
\end{definition}

\subsection{Last Species Standing}\label{sec:lastSpecies}

A similar analysis used for the predator-prey model
can be applied to yield the following result.

\begin{theorem}[Last Species Standing] \label{thm:lastspecies}
    Suppose $\vec \alpha_0 \ = \ (\alpha_0, \dots, \alpha_0)$ and
    $\vec \beta_0 \ = \  (\beta_0, \dots, \beta_0)$.
    In a Leslie competitive model, one of the two
    species is likely to vanish as $n \rightarrow \infty$. The fate of the two species is determined by the sign of the term
    \begin{equation}
        D \ :=\ \alpha_0 - \sqrt{m} \beta_0.
    \end{equation}
    In particular, if $D \ > \  0$, then the population $\alpha$ vanishes and population $\beta$ grows exponentially. If $D \ < \  0$, then the population $\beta$ vanishes and the population $\alpha$ grows exponentially. If $D \ = \  0$, either both species vanish or both grow exponentially.
\end{theorem}
\begin{proof}
    Proposition \ref{thm:to2ndOrd} can be generalized by the substitution $m \mapsto -m$. From the recursive relation
    \begin{equation}
        \alpha_n \ = \ (2L)\alpha_{n - 1} -L^2 \alpha_{n - 2} + mk^2 \alpha_{n - 2}
    \end{equation}
    we obtain the characteristic equation
    \begin{equation}
    \Lambda ^2 - 2L \Lambda + L^2 - mk^2 I \ =\ O,
    \end{equation}
    where $I, O$ are the identity matrix and the zero matrix of dimension $N$-by-$N$ respectively.
    By the quadratic formula, we derive the following roots:
    \begin{eqnarray}
        \Lambda_1  & \ = \ & L + k\sqrt{m} I
        \nonumber \\
        \Lambda_2 & \ = \ & L - k\sqrt{m} I,
    \end{eqnarray}
    Notice that $k$ and $m$ are both non-negative real values, and that $L$ is assumed to guarantee positive population growth. Hence,
    $\Lambda_1$ has a positive eigenvalue.

    From (\ref{eqn:MaxEvalGen}), we characterize the population as
    \begin{equation}\label{eqn:closedEq}
    \vec \alpha_n \ =\ \Lambda_1^{n}  \vec v_1+ \Lambda_2^{n}\vec v_2.
    \end{equation}
    In the limit as $n \rightarrow \infty$,
    \begin{equation}
    \vec \alpha_n \ \approx\ \Lambda_1^{n}  \vec v_1.
    \end{equation}
    Thus, the population is non-vanishing if and only if $\vec v_1$ is positive. We compute $\vec v_1$ directly. From \ref{eqn:closedEq}, we obtain two conditions:
    \begin{eqnarray}
        \vec \alpha_0 & \ = \ & \vec v_1 + \vec v_2 \nonumber \\
        \vec \alpha_1 & \ = \  & \Lambda_1 \vec v_1 + \Lambda_2 \vec v_2.
    \end{eqnarray}
    Solving for $\vec v_1$ yields
    \begin{equation}
        \begin{split}
        \vec v_1 \ &\ = \ \frac {\Lambda_2 \vec\alpha_0 - \vec\alpha_1} {\Lambda_2 - \Lambda_1}
        \ = \
        \frac {L\vec \alpha_0 - k\sqrt m \vec \alpha_0 - L \vec \alpha_0 + km\beta_0} {2mk}
        \ = \
        \frac {\sqrt{m}\vec\beta_0 -  \vec \alpha_0} {2\sqrt m} \\
        & \ = \ \frac {\sqrt m \beta_0 - \alpha_0} {2 \sqrt m} (1, \dots, 1)
        \ = \ -\frac D{2\sqrt m} (1, \dots, 1).
        \end{split}
    \end{equation}
    Similarly, we obtain
    \begin{equation}
        \beta_n \ \approx \ \Lambda_1 \vec w_1,
    \end{equation}
    where
    \begin{equation}
       \vec w_1 \ = \ \frac D {2\sqrt m }(1, \dots, 1).
    \end{equation}
    If $D \neq 0$, substituting the appropriate value of $D$ yields the desired result. Now suppose $D \ = \  0$ or $\alpha_0 \ = \ \sqrt m \beta _0$, then the conditions of \ref{eqn:conditions} imply
    \begin{eqnarray}
        \vec \alpha_1  \ = \ L \vec \alpha_0 - k \sqrt m \vec \alpha_0 \nonumber \\
        \vec \beta_1 \ = \ L \vec \beta_0 - k \sqrt m \vec \beta_0.
    \end{eqnarray}
    By induction, it is possible to prove that
    \begin{eqnarray}
        \vec \alpha_n \ = \ \sqrt m \vec \beta_n
    \end{eqnarray}
    for all nonnegative integers $n$. In turn, we obtain
    \begin{eqnarray}
        \vec \alpha_{n + 1} \ &=& \
        \left(
            L - k \sqrt m
        \right)^n \vec \alpha_0 \nonumber
        \\
        \vec \beta_{n + 1} \ &=& \
        \left(
            L - k \sqrt m
        \right)^n \vec \beta_0,
    \end{eqnarray}
    and the two populations grow or vanish simultaneously.
\end{proof}

\section{Applications of the Competitive Model: Limitation of Resources and Machine Learning}\label{sec:machinelearning}

In classical ecology, the competitive exclusion principle predicts that two
populations that compete over the same resources are unlikely to coexist. The
last species standing theorem supports this principle. We apply the model to
existing data on the population dynamics of species, and
quantify the competitive advantage of one species over the other.

Though unlikely, competitive coexistence is possible in the case of $D \ = \  0$.
We study the dynamics at the equilibrium point of competitive coexistence, and demonstrate its coherence with P. H. Leslie's original observations.

\subsection{Adjusting for Limited Resources}
Resources are limited in a realistic biological system, and it is
impossible that both species display indefinite exponential growth. We assume that the both populations of the model rely on the same resource with a logistical growth rate.

Let $r_n(t):\mathbb R \rightarrow \mathbb [0, 1]$ be a continuous function such that $r_n(t)$
describes the population density of the resource at the time interval $[(n - 1)\tau, n\tau)$. $r_n(t) \ = \  1$ implies maximum resource and $r_n(t) \ = \  0$ implies absence of resource. Though $r_n(t)$ models the continuous population evolution in between the discrete timestep of the competitive model, our major concern is the value of $r_n(0)$ and $r_n(\tau)$ where $\tau$ is the unit time for one instance of \eqref{eqn:conditions}.

Let $T$ be the total population capacity of the system. If the combined population of the competitive species equals the total capacity, the resources are totally consumed. The consumption happens when the time is an integer multiple of $\tau$.
We deduce the initial condition and the logistic growth that allows the resource population to asymptote to 1. If the total population exceeds the capacity, the
resources go extinct, and $r(t) \ = \  0$. By the observation, the equation for $r_{n+1}(t)$ is given as
\begin{align}
    r_{n + 1}(0) & \ = \ \max\left(1 - \frac {\alpha_n + \beta_n} T, 0\right) \\
    r_{n + 1}(t) & \ = \ \frac{\max(T - \alpha_n - \beta_n, 0) e^{mt}} {
    \max(T - \alpha_n - \beta_n, 0) e^{mt} + \alpha_n + \beta_n
    }.
\end{align}
The quantity $m$ determines the replenishing rate of the resources. For high values of $\tau$, we can estimate the $r_n(t)$ as the following Iverson bracket, i.e.,
\begin{align}
    r_{n + 1}(\tau) \ = \ [T \ > \  \alpha_n + \beta_n].
\end{align}
To justify the claim, we plot the value of $r(\tau)$ for varying
values of $\alpha_n+\beta_n$ in figure \ref{fig:rtauplots}.

\begin{figure}[h]
    \centering
    \begin{subfigure}[b]{0.45\textwidth}
        \includegraphics[width=\textwidth]{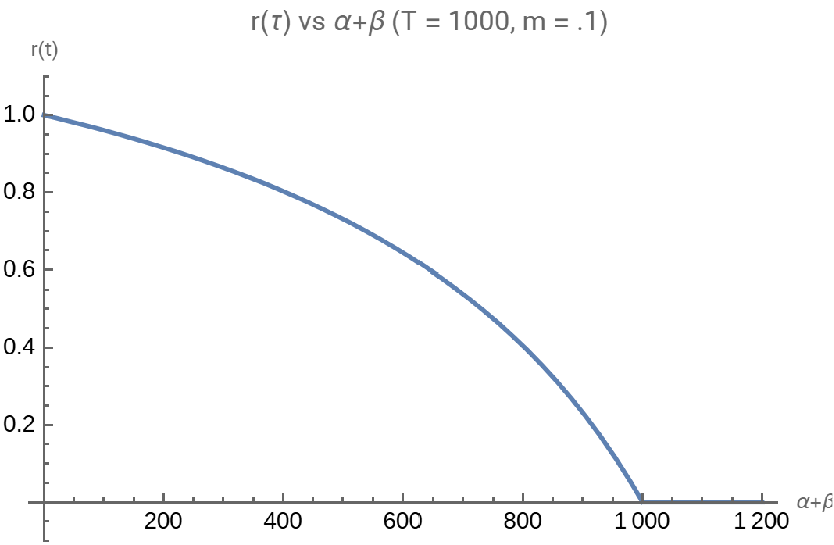}
        \label{fig:fig1}
    \end{subfigure}
    \hfill
    \begin{subfigure}[b]{0.45\textwidth}
        \includegraphics[width=\textwidth]{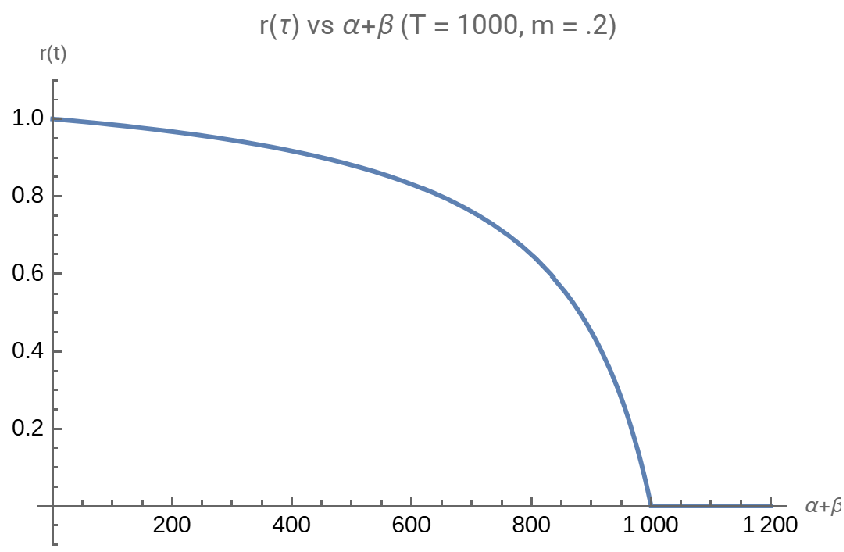}
        \label{fig:fig2}
    \end{subfigure}
    \
    \begin{subfigure}[b]{0.45\textwidth}
        \includegraphics[width=\textwidth]{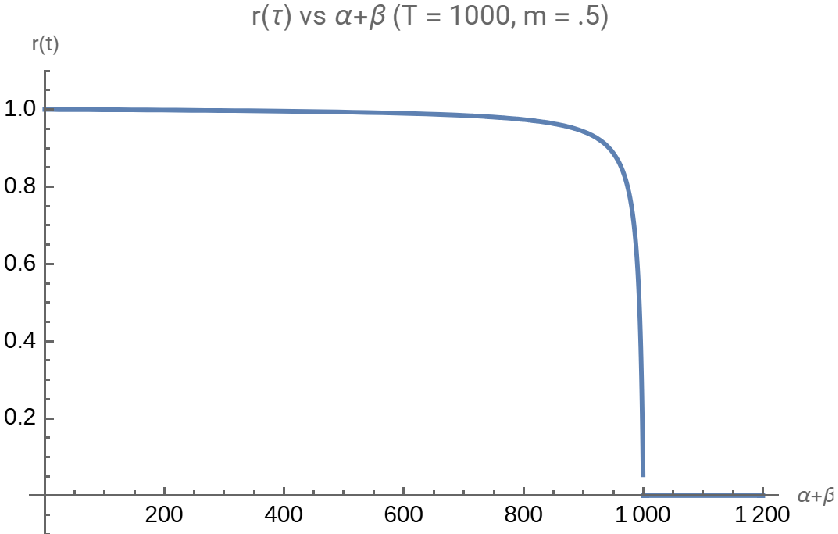}
        \label{fig:fig3}
    \end{subfigure}
    \hfill
    \begin{subfigure}[b]{0.45\textwidth}
        \includegraphics[width=\textwidth]{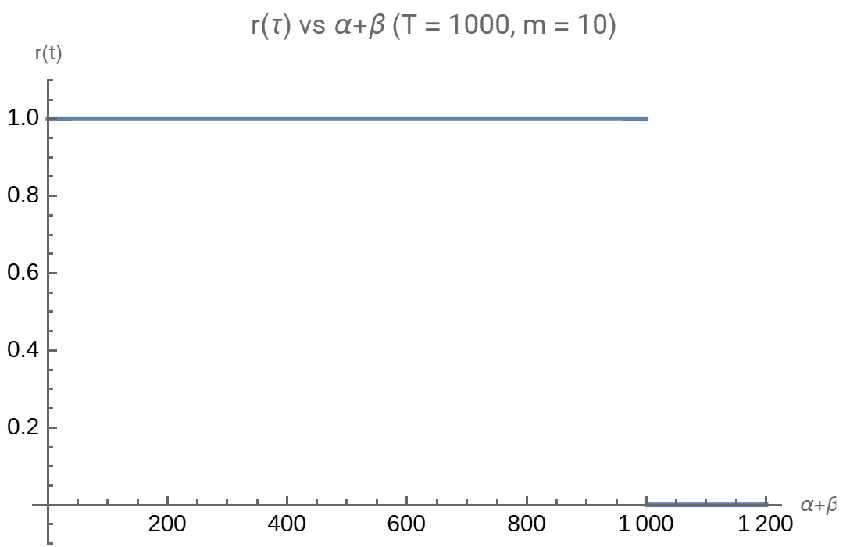}
        \label{fig:fig4}
    \end{subfigure}
    \caption{Plot of $r(\tau)$ versus $\alpha_n + \beta_n$. }
    \label{fig:rtauplots}

\end{figure}

Since the resources are exhausted when the population goes over
the population capacity $T$, we deduce that growth and competition stops when $\alpha_n + \beta_n \ > \  T$. The natural death of the old age groups allow the resources to replenish \cite{Bag19}. As a consequence, a small perturbation is introduced around the equilibrium. To simplify our calculations, we ignore the perturbation at equilibrium and claim that the population stagnates once the population reaches its full capacity.

\subsection{Competitive Coexsistence and Machine Learning} \label{sec:ML}

A classical example of competitive coexistence is from the experiment of Gause in the 1930's. The population of Paramecium Aurelia and Paramecium Caudatum\footnote{P. Aurelia and P. Caudatum in short.} was observed under a controlled environment. Both species shared a resource of Bacillus pyocyaneus. The combined population of the two species display exponential growth until the equilibrium population is reached. At equilibrium, the two species coexist. P. Aurelia possesses competitive advantage, and therefore has a faster growth rate along with a higher equilibrium population.

P. H. Leslie studied the population dynamics using the Lotka-Euler equations. We carry out a similar analysis using the Leslie model.
In order to account for limited resources, we focus on the data before equilibrium is attained. In the experimental setup of Gause, the liquid medium containing nutrients were replaced everyday after the observation. Therefore, if the resources consumed per day by the combined population does not reach the threshold, the effect of the limitation can be ignored.

For a quantitative analysis, we consider the specifics of the Competitive model. Fix the number of population groups to be $N \ = \  3$. The population dynamics depend on five quantities
\begin{align}(f, k, \alpha_0, \beta_0, m)\end{align}
where each parameter describes fertility, predation, initial population of the dominant population, initial population of the subdominant population, and the multiplier. To ensure that $\vec \alpha$ is the
dominant population, we set $m \ > \  1$.The recursive equation for the competitive model is given from \eqref{eqn:conditions}.

Since we assume competitive coexistence, $D \ = \  0$ by theorem \ref{thm:lastspecies}. Therefore, we fix the initial population of the
dominant species as
\begin{align}
    \alpha_0 \ = \ \sqrt{m}\beta_0.
\end{align}

Standard regression techniques fail to provide a fit of the theoretical model. The challenge is that the resulting curve of the population does not have an exact analytic expression, but rather is defined by a recursive relation determined by the parameter. In order to overcome the increasing complexity for larger values of $T$, we draw to techniques from machine learning.

Let $p, q:[T]\rightarrow \mathbb R$ be two functions that each describe
the observed population of the dominant and subdominant species\footnote{$[T] \ := \ \{1, 2, \dots, T\}$ describes the discrete timestep. For Gause's experiment, $t\in[T]$ describes the days elapsed from the beginning of the experiment. }. Also, define
\begin{align}
    \overline\alpha_t \ = \ \sum_{i \in [N]} \vec \alpha_t \cdot e_i
    \textAnd
    \overline \beta_t \ = \ \sum_{i \in [N]} \vec \beta_t \cdot e_i.
\end{align}
In this definition, $\vec \alpha_t, \vec \beta_t$ is the population vector at time $t$ of a competitive model with parameters $(f, k, \beta_0, m)$. The vectors $e_i \ := \ (0, \dots, 0, 1, 0, \dots, 0)$ are the $i$th
cannonical bases.

In order to fit the model, we define a function $\chi: \mathbb R^4 \rightarrow \mathbb R$
that describes the fit.

\begin{align}
    \chi(f, k, \beta_0, m) \ := \
\left(\sum_{t \ \leq \  T} \left(\frac{\overline \alpha_t - p(t)} {p(t)}\right)^2 +
    \sum_{t \ \leq \  T} \left(\frac{\overline \beta_t - q(t)} {q(t)}\right)^2
    \right)
\end{align}
Also fix the value of the number of populations to be $N \ = \  3$.

By the nature of the model, the ratio between the equilibrium population depends on $m$, and the growth of both populations are determined by $f$. The system displays extreme sensitivity to the two values $k, \beta_0$. The two values $k, \beta_0$ are fixed after they enter an admissible range.

We performed the following machine learning scheme to train the model to extract the parameters that minimizes $\chi$.

\vspace{2mm}

\underline{\textsc{Machine Learning Scheme}}
\begin{enumerate}
    \item Evaluate $\chi$ for random choice of the four parameters.
    Among the random tuples, choose the tuple that minimizes the value of $\chi$.
    \item Perform a general gradient descent involving all four parameters, starting from the point chosen in step 1. Find the admissible value of $k, m$.
    \item Define a subroutine, $\rm{Opt_g}$, that performs a gradient descent solely on parameter $g$
    with all other parameters fixed.
    \item Finally perform an optimized gradient descent on parameter $m$. In the beginning of each loop, invoke the subroutine $\rm{Opt_g}$ to adjust the parameter $g$.
\end{enumerate}

\vspace{2mm}

Gause's original data includes population evolution over 16 days where equilibrium population is attained in day 10. We set $T \ = \  10$ and tested for 20 parameters. The threshold value of admissible $\chi$ for a 5\% significance and 20 degrees of freedom is
is 10.851 \cite{NIS22}.

After training the model according to the scheme described above, we obtained a fit with $\chi \ = \  3.11$.

\begin{figure}[htp]
    \centering
    \includegraphics[width=0.8\textwidth]{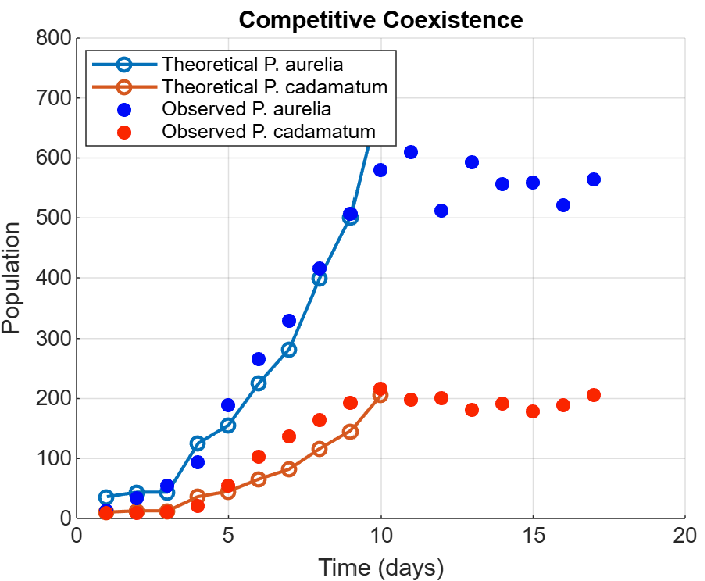} 
    \caption{Fit of the competitive model to experemental data. Fit parameters are $(f, \beta_0, k, m) \ = \ (1.79, 0.56, 3.53, 9.42
)$, and the associated error is $\chi(f, \beta_0, k, m) \ = \ 3.11$. }
    \label{fig:example}
\end{figure}

\section{Generalised Framework: Quantum Operators}

\subsection{Motivation}
From Equation \ref{eqn:MaxEvalGen}, we have already seen that a complex model approach allows us to extract a lot more information compared to a real model approach, where we have to work with more rigid closed-form solutions. Using this as a motivation, we now propose a quantum mechanical approach for modeling the time evolution of populations with discrete age demographics, by using bosonic ladder operators.

The application of ladder operators, or more colloquially creation and annihilation operators, from quantum mechanics to model complex real-world interactions between $N$ systems is well-studied \cite{Bag19}
. Previously, a fermionic ladder operator approach was favored as it allowed for discrete evolutionary states, but in recent times, a truncated bosonic approach has been developed with great success, to study population evolution \cite{ArG89}. We use previous work on this subject as a motivation to develop a bosonic ladder operator approach to study populations wtih discrete age-structures and to relate it to the Leslie matrix approach from our previously developed model.

\subsection{Hermitian of the Leslie Matrix}
In a quantum system, the time evolution of a state $\ket{\varphi}$ can be described by the Schrodinger Equation:

\begin{equation}
    \frac {\partial} {\partial t}i \hbar \ket{\varphi(t)} \ = \
    i H \ket{\varphi(t)}.
\end{equation}

The solution to this equation can be described by the exponential map:

\begin{equation}
    \ket{\varphi(t)} \ = \ e^{iHt}\ket{\varphi(0)}.
\end{equation}

A population model can also be described using quantum operators. In this case, the population vector becomes the state. We omit the normalization constant $\hbar$ for better notation.

\begin{prop}[Hamiltonian for a Single-Population Leslie Matrix]
    Consider a single non-interacting population that follows the time-evolution dictated by the Leslie matrix $L$. If 
    \begin{align}
    \|L - I\| \ \leq \  1,
    \end{align}
    where $I$ denotes the identity operator, then our Hamiltonian has a closed form solution:
    \begin{align}H \ = \ \frac{(-1)^k}{ik} (L - I)^k.\end{align}
\end{prop}

\begin{prop}[Algorithm to Compute the Hamiltonian]
     The matrix logarithm of $L$ is
    \begin{eqnarray}
        \log(L) \ = \ P \log(D) P^{-1},
    \end{eqnarray}
    where
    \begin{eqnarray}
        D \ = \ \textnormal{diag}( d_1, \dots, d_N)
    \end{eqnarray}
    and
    \begin{equation}
        \log(D) \ = \ \textnormal{diag}( \log(d_1), \dots, \log(d_N)).
    \end{equation}
\end{prop}

\begin{proof}
    We directly take the exponential and show that it matches $L$.
    \begin{eqnarray}
        e^{\log(L)} \ = \
        \sum_{k \ = \  0}^\infty
        \frac {(\log(L))^k} {k!}\ = \
        \sum_{k \ = \  0}^\infty
        \frac {(P \log(D) P^{-1})^k} {k!}\ = \
        \sum_{k \ = \  0}^\infty
        \frac {P \log(D)^k P^{-1}} {k!} \nonumber
        \\
        = \ P\left( \sum_{k \ = \  0}^\infty
        \frac {\log(D)^k}{k!}\right) P^{-1} \ = \ PDP^{-1} \ = \ L
    \end{eqnarray}

We also provide an example use of the algorithm. We consider the population model described in the introduction (Equation \ref{eqn:motivatingModel}) and set the parameters to be
$(f, m, k, F) \ = \ (.2, .5, .5, 2)$. Numerically, the model simplifies to

\begin{eqnarray}\label{eqn:motivatingModelEx}
    \vec p_{n + 1} \ :=\ \widetilde L\vec p_n \ = \
    \begin{bmatrix}
        .2 & .2 & .2 & .5\\
        1 & 0 & 0 & 0\\
        0 & 1 & 0 & 0 \\
        -.5 & -.5 & -.5 & 3
    \end{bmatrix}\vec p_n.
\end{eqnarray}

We obtain
\begin{equation}
    \log\left(\widetilde L\right) \ \approx \
    \begin{bmatrix}
-0.3926 & -0.3606 & 0.5079 & 0.2940 \\
1.9656 & -0.9006 & -0.8685 & -0.2296 \\
-3.2005 & 2.8342 & -0.0320 & 0.4569 \\
-0.5212 & -0.0644 & -0.2940 & 1.1627
\end{bmatrix},
\end{equation}

and by direct computation via MATLAB, we verify that

\begin{equation}
    e^{\log(L)} \ \approx \  \begin{bmatrix}
        .2 & .2 & .2 & .5\\
        1 & 0 & 0 & 0\\
        0 & 1 & 0 & 0 \\
        -.5 & -.5 & -.5 & 3
    \end{bmatrix} \ = \ L.
\end{equation}

\end{proof}

 \appendix

\section{Roots of $\ch_N(z)$} \label{app:A}

In this section, we provide our observations on the roots of $\ch_N(z)$ using Complex Analysis.

\begin{theorem}
    If $f \ \geq \ 1/N$, then $\ch(z)$ has a unique positive, real
    root that has a magnitude strictly greater than any of the other
    complex roots.
\end{theorem}

\begin{proof}
    Consider the polynomial
    \begin{align}
        h(z) \ := \ (z - 1) \ch_N(z) \ = \ z^{N + 1} - (f+1) z^N + f
    \end{align}
    which has a simpler algebraic expression.
    We split the
    polynomial $h(x)$ into two summands, and invoke Rouche's Theorem (\cite{SS03} p91).
    Let $C_{1 + \epsilon}$ be a circular contour centered at the origin
    with radius $1 + \epsilon$ for arbitrarily small $\epsilon$.
    Write
    \begin{align}
        h(z) \ = \ (z^{N + 1} + f) + (1 + f)z^N
    \end{align}
    and Taylor expand the two summands at $z = 1$.
    \begin{align}
        (z^{N + 1} + f) &\ = \ 1 + f + (N + 1) \epsilon \\
        z^N(1 + f) &\ = \ (1 + N \epsilon)  (1 + f) \ = \ 1 + f + N(1 + f)\epsilon
    \end{align}
    By assumption, $f \ \geq \ 1/N$, which implies $(N + 1) \ \leq \ N(1 + f)$.
    The modullus of the two terms along the contour can be compared
    as follows.
    \begin{align}
        \left|z^{N + 1} + f\right| \ \leq \ \left|(1 + f)z^N\right|
    \end{align}
    By Rouche's theorem, $h(z)$ has the same number of roots as the
    term that has a larger modullus in the countour $C_{1 + \epsilon}$, which is the
    summand $(1-f)z^N$. It is trivial to see that this summand has $N$
    roots inside the countour, and by fundamental theorem of algebra,
    $h(z)$ has $N + 1$ roots.

    We know that $\ch_N(z)$ is positive somewhere in the interval $[1, \infty)$.
    We consider the following:
    \begin{equation}
        \ch_N(1) \ =\  1 - fN \ \leq \ 0.
    \end{equation}
    By the Intermediate Value Theorem, we conclude that the one root outside
    the unit circle is a positive real value.
\end{proof}

\begin{theorem}
    If $f \ < \ 1/N$, then all the roots of $\ch_N(z)$ have
    a modullus strictly less than $1$.
\end{theorem}

\begin{proof}
    It suffices to show that
    \begin{equation}
        \widetilde h(z) \ = \ h (1/z) z^{N + 1} \ = \ fz^{N + 1} - (f+1)z+1
    \end{equation} has exactly one root within the unit circle which
    comes from multiplying \newline $(z - 1)$. Again, consider the countour
    $C_{1+\epsilon}$ and split $\widetilde h (z)$ into two summands.
    \begin{align}
        \widetilde h(z) \ = \ (f z^{N + 1} + 1) - (f + 1)z
    \end{align}
    Taylor expand the two summands at $z \ = \ 1$, and notice that under the condition
    $f \ < \ 1/N$, the second summand has a larger modullus along the
    contour $C_{1 + \epsilon}$.
    \begin{align}
        fz^{N + 1} + 1 & \ = \ f(1 + (N + 1)\epsilon) + 1  \\
        (f + 1)z & \ = \ (f + 1) (1 + \epsilon)
    \end{align}
    Clearly, the second summand has one root inside the contour $C_{1 + \epsilon}$,
    which originates from $(z - 1)$. By Rouche's theorem, $\widetilde h(z)$
    has exactly one root inside the unit circle, i.e. $z = 1$, and all
    other roots have a modullus greater than 1. Consequently, $\ch_N(z) \ = \ h(z)/(z + 1)$
    has all of its roots strictly inside the unit circle.
\end{proof}

\begin{theorem}[Bounds for the largest root]\label{thm:Bound}

 Given that $f \ \geq \ 1/N$, the largest root $\lambda_{\max}$ of $\ch_N(z)$ given by
 \begin{equation}
 1 + f - \frac 1 {N} \ \leq\ \lambda_{\max} \ < \ 1 + f.
 \end{equation}
\end{theorem}

\begin{proof}
    The upper bound is trivial:
    \begin{equation}
        \ch_N(1+f) \ = \ f \ >\  0.
    \end{equation}
    We have $\ch_N(0) \ = \  -f \ < \  0$, and thus by the Intermediate Value Theorem the maximum root is
    bounded.

    To obtain the lower bound, we write $f \ =  \ 1/N + \epsilon$ for
    some $\epsilon \ \geq \  0$. With some algebra listed below
    , we compute $\ch_N(z)$ at the claimed lower
    bound. If we show that this value is less than zero, the dominating root must be greater than the purported lower bound. We find
    \begin{equation}
        \ch_N\left(
            1 + f - \frac 1 N
        \right)  \ = \ -
        \left(
            1 + f - \frac 1 N
        \right)^N \left[
            \frac {1} {fN - 1}
        \right]
        + \frac {fN} {fN - 1}.
    \end{equation}
    We wish to bound this value by zero. It suffices to show
    \begin{equation}
        fN - \left(
            1 + f - \frac 1 N
        \right)^N \ \leq\ 0,
    \end{equation}
    which, using the $\epsilon$ substitution, converts to
    \begin{equation}
        1 + N\epsilon - (1 + \epsilon)^N \ \geq \ 0.
    \end{equation}
    Expanding the power term by the binomial theorem, we see that inequality indeed holds.
\end{proof}

The following figures demonstrate the theorems in this section.

\begin{figure}[h]
    \centering
    \begin{subfigure}[b]{0.45\textwidth}
        \includegraphics[width=\textwidth]{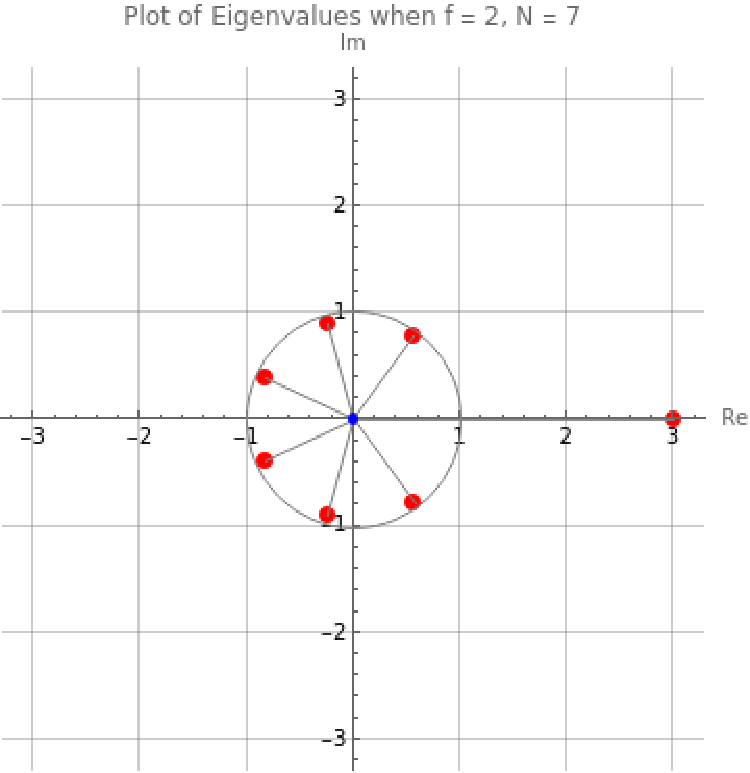}
        \caption{$f \ > \  1$, $N$ odd.}
        \label{fig:fig1}
    \end{subfigure}
    \hfill
    \begin{subfigure}[b]{0.45\textwidth}
        \includegraphics[width=\textwidth]{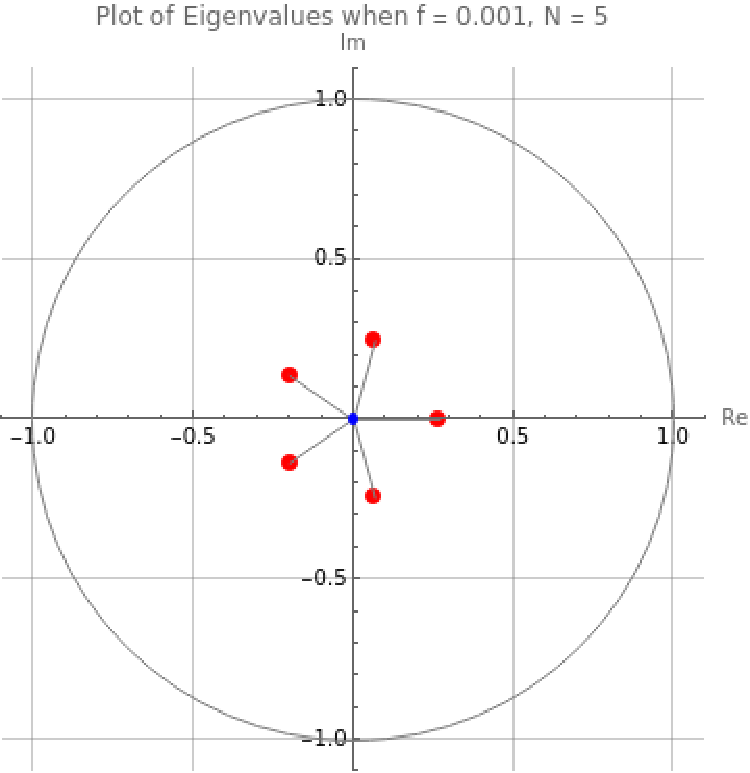}
        \caption{$f \approx 0$, $N$ odd.}
        \label{fig:fig2}
    \end{subfigure}
    \\
    \begin{subfigure}[b]{0.45\textwidth}
        \includegraphics[width=\textwidth]{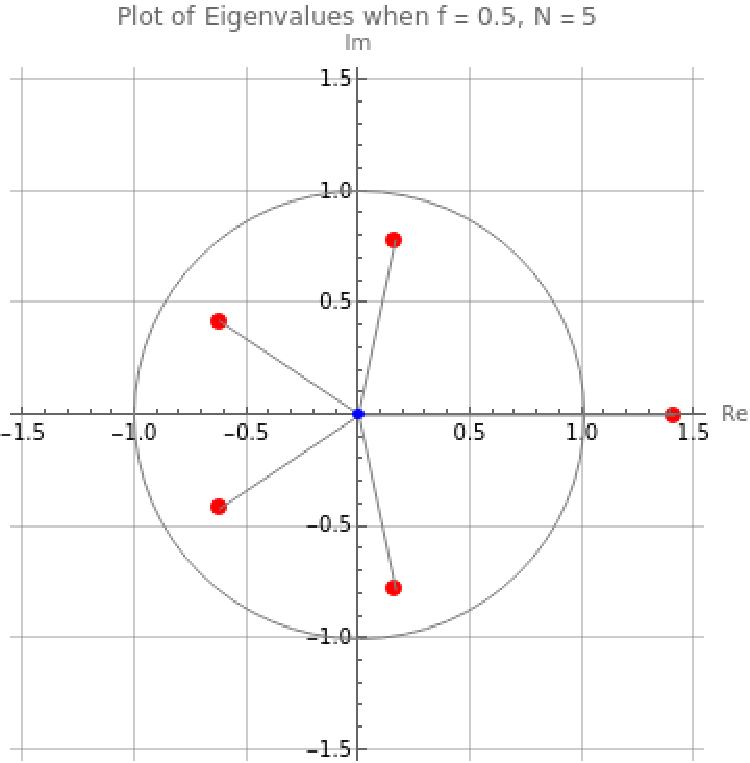}
        \caption{$f \ < \ 1$, $N$ odd.}
        \label{fig:fig3}
    \end{subfigure}
    \hfill
    \begin{subfigure}[b]{0.45\textwidth}
        \includegraphics[width=\textwidth]{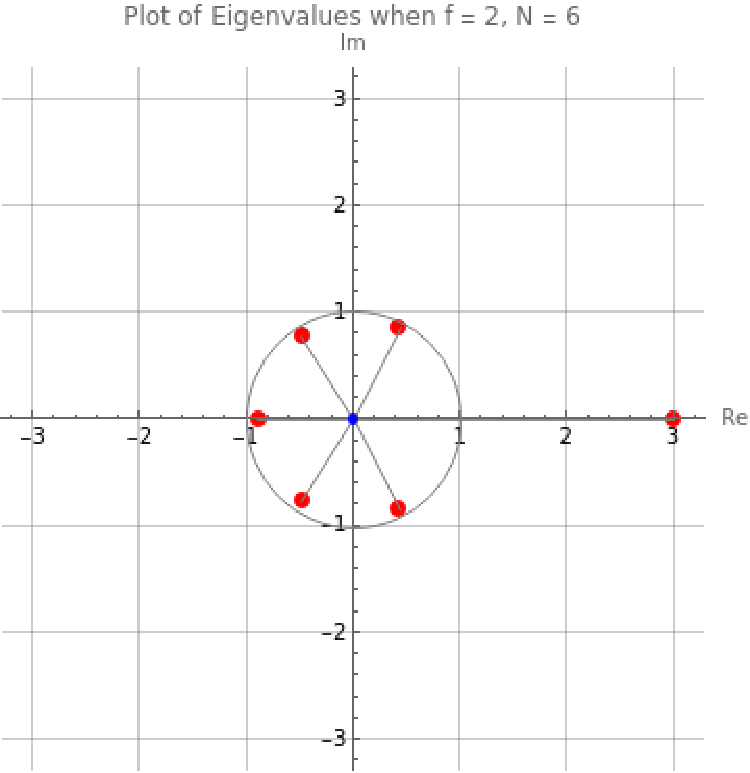}
        \caption{$f \ > \  1$, $N$ even.}
        \label{fig:fig4}
    \end{subfigure}
    \caption{Roots of $\ch_N(z)$ for varying $f, N$ in the complex plane.}
    \label{fig:overall}
\end{figure}



\ \\

\end{document}